\documentclass[10pt,a4paper,envcountsame]{llncs}

\usepackage[bookmarks,bookmarksopen,bookmarksdepth=2]{hyperref}

\usepackage{xspace}
\usepackage{comment}
\usepackage{amsmath}
\usepackage{amsfonts}
\usepackage{amssymb}
\usepackage{pdfsync}
\usepackage{thmtools,thm-restate}
\usepackage[textsize=small]{todonotes}
\usepackage{paralist}

\usepackage[textsize=small]{todonotes}
\newcommand{\lorenzo}[1]{\todo[color=green]{L: #1}}
\newcommand{\sla}[1]{\todo[color=green]{S: #1}}

\newcommand{\ignore}[1]{}

\newcommand{\F}{\mathcal{F}}
\newcommand{\G}{\mathcal{G}}

\newcommand{\N}{\mathbb{N}}
\newcommand{\Z}{\mathbb{Z}}

\newcommand{\set}[1]{\{#1\}}
\newcommand{\setof}[2]{\set{#1 \mid #2}}
\newcommand{\card}[1]{\left|#1\right|}
\newcommand{\prettyexists}[2]{\exists_{#1} \, #2}
\newcommand{\prettyforall}[2]{\forall_{#1} \, #2}
\newcommand{\source}[1]{\textsc{source}(#1)}
\newcommand{\target}[1]{\textsc{target}(#1)}
\newcommand{\runs}[1]{\textsc{Runs}(#1)}  
\newcommand{\runord}{\unlhd}
\newcommand{\fin}{\textsc{fin}}
\newcommand{\lin}{\textsc{Lin}}

\newcommand{\lab}{\textsc{label}}
\newcommand{\parimage}{\textsc{pi}}
\newcommand{\reachvas}{\textsc{Reach}(\mathrm{VAS})}
\newcommand{\secreachvas}{\textsc{SecReach}(\mathrm{VAS})}
\newcommand{\reachvass}{\textsc{Reach}(\mathrm{VASS})}
\newcommand{\secreachvass}{\textsc{SecReach}(\mathrm{VASS})}
\newcommand{\reach}{\textsc{Reach}}
\newcommand{\secreach}{\textsc{SecReach}}
\newcommand{\vas}{VAS\xspace}
\newcommand{\vases}{{\vas}es\xspace}

\newcommand{\vass}{VASS\xspace}
\newcommand{\vasses}{{\vass}es\xspace}

\renewcommand{\sec}[2]{\textsc{sec}_{#1}(#2)}
\newcommand{\proj}[2]{\pi_{#1}(#2)}
\newcommand{\cc}{\textsc{cc}}
\newcommand{\trans}[1]{\stackrel{#1}{\longrightarrow}}

\newcommand{\Nplus}{\N_{\geq 0}}	
\newcommand{\zeroel}[1]{0_{#1}}

\newcommand{\sepsep}{ }

\newcommand{\vasreach}{\vas reachability\xspace}

\newcommand{\vaslangs}{\vas languages\xspace}

\renewcommand{\enspace}{}

\begin{document}

\title{Separability of Reachability Sets of Vector Addition Systems}

\author{Lorenzo Clemente\inst{1}
\and Wojciech Czerwi\'nski\inst{1}
\and S{\l}awomir Lasota\inst{1}
\and Charles Paperman\inst{2}}
\institute{University of Warsaw \and University of T{\"u}bingen}

\maketitle

\begin{abstract}
	Given two families of sets $\F$ and $\G$,
	the $\F$ separability problem for $\G$ asks whether for two given sets $U, V \in \G$ there exists
	a set $S \in \F$, such that $U$ is included in $S$ and $V$ is disjoint with $S$.
	We consider two families of sets $\F$:
	modular sets $S \subseteq \N^d$, defined as unions of equivalence classes modulo some natural number $n \in \N$,
	and unary sets.
	Our main result is decidability of modular\sepsep and unary{\sepsep}separability for the class $\G$ of reachability sets of Vector Addition Systems,
	Petri Nets, Vector Addition Systems with States, and for sections thereof.
\end{abstract}


\section{Introduction}
In this paper we mainly investigate separability problems for sets of vectors from $\N^d$.
We say that a set $U$ is \emph{separated from} set $V$ by a set $S$ if $U \subseteq S$ and $V \cap S = \emptyset$.
For two families of sets $\F$ and $\G$,
the \emph{$\F$-separability problem for $\G$}
asks for two given sets $U, V \in \G$
whether $U$ is separated from $V$ by some set from $\F$.
Concretely, we consider $\F$ to be modular sets or unary sets,
and $\G$ to be reachability set of Vector Addition Systems, or generalizations thereof.

\paragraph{Motivation.}
The separability problem is a classical problem in theoretical computer science.
It was investigated most extensively in the area of formal languages, for $\G$ being the family of all regular word languages.
Since regular languages are effectively closed under complement, the $\F$-separability problem
is a generalization of the $\F$-characterization problem, which asks whether a given language belongs to $\F$.
Indeed, $L \in \F$ if and only if $L$ is separated from its complement by some language from $\F$.
Separability problems for regular languages attracted recently a lot of attention,
which resulted in establishing the decidability of $\F$-separability for the family $\F$ of separators being
the piecewise testable languages~\cite{DBLP:conf/icalp/CzerwinskiMM13,DBLP:conf/mfcs/PlaceRZ13}
(recently generalized to finite ranked trees \cite{Schmitz:Goubault-Larrecq:ICALP2016}),
the locally and locally threshold testable languages~\cite{DBLP:conf/fsttcs/PlaceRZ13},
the languages definable in first order logic~\cite{DBLP:journals/corr/PlaceZ14},
and the languages of certain higher levels of the first order hierarchy~\cite{DBLP:conf/icalp/PlaceZ14},
among others.

Separability of nonregular languages attracted little attention till now.
The reasons for this are twofold. First, for regular languages one can use standard algebraic tools, like syntactic monoids,
and indeed most of the results have been obtained with the help of such techniques.
Second, some strong intractability results have been known already since 70's, when Szymanski and Williams proved
that regular{\sepsep}separability of context-free languages is undecidable~\cite{DBLP:journals/siamcomp/SzymanskiW76}.
Later Hunt~\cite{DBLP:journals/jacm/Hunt82a} generalized this result: he showed that $\F$-separability of context-free languages
is undecidable for every class $\F$ which is closed under finite boolean combinations
and contains all languages of the form $w\Sigma^*$ for $w \in \Sigma^*$. This is a very weak condition, so it seemed
that nothing nontrivial can be done outside regular languages with respect to separability problems.
Furthermore, Kopczy\'{n}ski has recently shown that regular{\sepsep}separability is undecidable
even for languages of visibly pushdown automata~\cite{DBLP:journals/corr/Kopczynski15a},
thus strengthening the result by Szymanski and Williams.
On the positive side, piecewise testable{\sepsep}separability has been shown decidable
for context-free languages, languages of Vector Addition Systems (\vaslangs),
and some other classes of languages~\cite{DBLP:conf/fct/CzerwinskiMRZ15}.
This inspired us to start a quest for decidable cases beyond regular languages.

To the best of our knowledge, beside~\cite{DBLP:conf/fct/CzerwinskiMRZ15},
separability problems for \vaslangs have not been investigated before.

%
%

\paragraph{Our contribution.}
In this paper, we make a substantial step towards solving regular{\sepsep}separability of \vaslangs.
Instead of \vaslangs themselves (i.e., subsets of $\Sigma^*$),
in this paper we investigate their commutative closures,
or, alternatively, subsets of $\N^d$ represented as reachability sets of \vases, \vases with states, or Petri nets.
A \vasreach set is just the set of configurations of a \vas which can be reached from a specified initial configurations.
Towards a unified treatment, instead of considering separately \vases, \vases with states, and Petri nets,
we consider \emph{sections} of \vasreach sets (abbreviated as \vas{\sepsep}sections below),
which turn out to be expressive enough to represent sections of \vases with states and Petri nets,
and thus being a convenient subsuming formalism.
A \emph{section} of a set of vectors $X\subseteq \N^d$ is the set obtained by first fixing a value for certain coordinates,
and then projecting the result to the remaining coordinates.
For example, if $X$ is the set of pairs $\setof {(x, y)\in \N^2}{x \textrm{ divides } y}$,
then the section of $X$ obtained by fixing the first coordinate to $3$ is the set $\set {0, 3, 6, \dots}$.
It can be easily shown that \vas sections are strictly more general than \vasreach sets themselves,
and they are equiexpressive with sections of \vases with states and Petri nets.

We study the separability problem of \vas sections by simpler classes, namely, modular and unary sets.
A set $X\subseteq \N^d$ is \emph{modular} if there exists a modulus $n\in\N$ s.t. $X$ is closed under the congruence modulo $n$ on every coordinate,
and it is \emph{unary} if there exists a threshold $n \in \N$ s.t. it is closed under the congruence modulo $n$ above the threshold $n$ on every coordinate.
Clearly, \vas sections are more general than both unary and modular sets,
and unary sets are more general than modular sets.
Moreover, unary sets are tightly connected with commutative regular languages,
in the sense that the Parikh image%
\footnote{The Parikh image of a language of words $L \subseteq \{a_1, \dots, a_k \}$ is the subset of $\N^k$ obtained by counting occurrences of letters in $L$.}
of a commutative regular language is a unary set, and vice versa,
the inverse Parikh image of a unary set is a commutative regular language.
As our main result,
we show that the modular and unary{\sepsep}separability problems are decidable for \vas sections (and thus for sections of \vases with states and Petri nets).
Both proofs use similar techniques,
and invoke two semi-decision procedures:
the first one (positive) enumerates witnesses of separability,
and the second one (negative) enumerates witnesses of nonseparability.
A separability witness is just a modular (or unary) set,
and verifying that it is indeed a separator easily reduces to the \vasreach problem.
Thus, the hard part of the proof is to invent a finite and decidable witness of nonseparability,
i.e., a finite object whose existence proves that none of infinitely many modular (resp.~unary) sets is a separator.
Our main technical observation is that two nonseparable \vasreach sets always admit
two \emph{linear} subsets thereof that are already nonseparable.

%

From our result, thanks to the tight connection between unary sets and commutative regular languages mentioned above,
we can immediately deduce decidability of regular{\sepsep}separability for \emph{commutative closures of \vaslangs},
and \emph{commutative regular}{\sepsep}separability for \vaslangs.
This constitutes a first step towards determining the status of regular{\sepsep}separability for languages of \vases.


\paragraph{Related research.}
Choffrut and Grigorieff have shown decidability of
separability of rational relations by recognizable relations in $\Sigma^* \times \N^d$~\cite{DBLP:journals/ipl/ChoffrutG06}.
Rational subsets of $\N^d$ are precisely the semilinear sets,
and recognizable (by morphism into a monoid) subsets of $\N^d$ are precisely the unary sets.
Thus, by ignoring the $\Sigma^*$ component,
one obtains a very special case of our result,
namely decidability of the unary{\sepsep}separability problem for semilinear sets.
Moreover, since modular sets are subsets of $\N^d$ which are recognizable by a morphism into a monoid which happens to be a group,
we also obtain a new result, namely,
decidability of separability of rational subsets of $\N^d$ by subsets of $\N^d$ recognized by a group.

From a quite different angle, our research seems to be closely related to the \vasreach problem.
Leroux~\cite{DBLP:conf/lics/Leroux09} has shown a highly nontrivial result:
the reachability sets of two \vases are disjoint if, and only if, they can be separated by a semilinear set.
In other words, semilinear{\sepsep}separability for {\vas} reachability sets is equivalent to the \vas (non-)reachability problem. 
This connection suggests that modular and unary{\sepsep}separability are interesting problems in themselves, enriching our understanding of \vases.
Finally, we show that \vasreach reduces to unary{\sepsep}separability,
thus the problem does not become easier by considering the simpler class of unary sets as opposed to semilinear sets.
For modular{\sepsep}separability we have a weaker complexity lower bound, i.e. \textsc{ExpSpace}-hardness,
by a reduction from control state reachability for \vasses.


\section{Preliminaries}\label{sec:prelim}

\paragraph{Vectors.}

By $\N$ and $\Z$ we denote the set of natural and integer numbers, respectively.
For a vector $u = (u_1, \ldots, u_d) \in \Z^d$ and for a coordinate $i \in \{1, \ldots, d\}$,
we denote by $u[i]$ its $i$-th component $u_i$.
The zero vector is denoted by $0$.
The order $\leq$ and the sum operation $+$ naturally extend to vectors pointwise.
Moreover, if $n \in \Z$, then $nu$ is the vector $(nu_1, \dots, nu_d)$.
These operations extend to sets element-wise in the natural way:
For two sets of vectors $U, V \subseteq \Z^d$ we denote by $U + V$ its Minkowski sum $\{u+v \mid u \in U, v \in V\}$.
For a (possibly infinite) set of vectors $S \subseteq \Z^d$, let $\lin(S)$ and $\lin^{\geq 0}(S)$
be the set of \emph{linear combinations} and \emph{non-negative linear combinations} of vectors from $S$, respectively, i.e.,
\begin{align*}
	\lin(S) 			&= \{ a_1 v_1 + \ldots + a_k v_k \mid v_1, \dots, v_k \in S, a_1, \dots, a_k \in \Z \}\enspace, \textrm{ and } \\
	\lin^{\geq 0}(S)	&= \{ a_1 v_1 + \ldots + a_k v_k \mid v_1, \dots, v_k \in S, a_1, \dots, a_k \in \N \}\enspace.
\end{align*}
When the set $S = \set{v_1, \dots, v_k}$ is finite,
we alternatively write $\lin(v_1, \dots, v_k)$ instead of $\lin(\set{v_1, \dots, v_k})$,
and similarly for $\lin^{\geq 0}(v_1, \dots, v_k)$.

\paragraph{Modular, unary, linear, and semilinear sets.}

Two vectors $x, y \in \Z^d$ are \emph{$n$-modular equivalent}, written $x \equiv_n y$,
if, for all $i \in \{1, \ldots, d\}$, we have $x[i] \equiv y[i] \mod n$.
Moreover, two non-negative vectors $x, y \in \N^d$ are \emph{$n$-unary equivalent},
written $x \cong_n y$, if $x \equiv_n y$ and $x[i] \geq n \iff y[i] \geq n$ for all $i \in \{1, \ldots, d\}$.
A $d$-dimensional set $S \subseteq \N^d$ is \emph{modular} if there exists a number $n \in \N$,
s.t. $S$ is a union of $n$-modular equivalence classes.
\emph{Unary} sets $S\subseteq \N^d$ are defined similarly w.r.t. $n$-unary equivalence classes.

A set $S \subseteq \N^d$ is \emph{linear}
if it is of the form $S = \{b\} + \lin^{\geq 0}(p_1, \ldots, p_k)$ for
some \emph{base} $b \in \N^d$ and some \emph{periods} $p_1, \ldots, p_k \in \N^d$.
A set is \emph{semilinear} if it is a finite union of linear sets.
Note that a modular set is also unary (since $\cong_n$ is finer than $\equiv_n$),
and that unary set is in turn a semilinear set,
which can be presented as a finite union of linear sets in which all
the periods are parallel to the coordinate axes,
i.e., they have exactly one non-zero entry.

\paragraph{Separability.}
For $S, U, V \subseteq \N^d$, we say that $S$ \emph{separates} $U$ from $V$ if $U \subseteq S$ and $V \cap S = \emptyset$.
The set $S$ is also called a \emph{separator} of $U,V$.
For a family $\F$ of sets, we say that $U$ is $\F${\sepsep}separable from $V$
if $U$ is separated from $V$ by a set $S \in \F$.
In this paper, the set of separators $\F$ will be the modular sets and the unary ones.
Since both classes are closed under complement, the notion of $\F${\sepsep}separability is symmetric:
$U$ is $\F${\sepsep}separable from $V$ iff $V$ is $\F${\sepsep}separable from $U$.
Thus we use also a symmetric notation, in particular we say that $U$ and $V$ are $\F${\sepsep}separable instead
of saying that $U$ is $\F${\sepsep}separable from $V$.
For two families of sets $\F$ and $\G$,
the \emph{$\F${\sepsep}separability problem for $\G$}
asks whether two given sets $U,V \in \G$ are $\F${\sepsep}separable.
In this paper we mainly consider two instances of $\F$, namely modular sets and unary sets, and thus we speak of 
\emph{modular{\sepsep}separability} and \emph{unary{\sepsep}separability} problems, respectively.

\paragraph{Vector Addition Systems.}

A $d$-dimensional \emph{Vector Addition System} (\vas) is a pair $V = (s, T)$, where $s \in \N^d$
is the \emph{source} configuration and $T \subseteq_\fin \Z^d$ is the set of finitely many \emph{transitions}.
A \emph{partial run} $\rho$ of a \vas $V = (s, T)$ is a sequence 
$$(v_0, t_0, v_1), (v_1, t_1, v_2), \ldots, (v_{n-1}, t_{n-1}, v_n) \in \N^d \times T \times \N^d$$ 
such that for all $i \in \{0, \ldots, n-1\}$ we have $v_i + t_i = v_{i+1}$. The 
\emph{source} of this partial run is the configuration $v_0$ and the 
\emph{target} of this partial run is the configuration $v_n$, we write 
$\source{\rho} = v_0$, 
$\target{\rho} = v_n$.
The \emph{labeling} of $\rho$ is the sequence $t_0 \ldots t_{n-1} \in T^*$, we write $\lab(\rho) = t_0 \ldots t_{n-1}$.
For a sequence $\alpha \in T^*$ and a partial run $\rho$ such that $\lab(\rho) = \alpha$, $\source{\rho} = u$ and $\target{\rho} = v$
we write $u \trans{\alpha} v$ to denote this unique partial run.
A partial run $\rho$ of $(s, T)$ with $\source{\rho} = s$ is called a \emph{run}.
The set of all runs of a \vas $V$ is denoted as $\runs{V}$.
The \emph{reachability set} $\reach(V)$ of a \vas $V$ is the set of targets of all its runs; the sets $\reach(V)$
we call \emph{\vasreach sets} in the sequel.
The family of all \vasreach sets we denote as $\reach(\text{\vas})$.
\begin{example}\label{ex:vasreach}
Consider a \vas $V = (s, T)$, for a source configuration $s = (1, 0, 0)$
and a set of transitions $T = \{(-1, 2, 1), (2, -1, 1)\}$.
One easily proves that
\[
\reach(V) \ = \ \setof{(a, b, c) \in \N^2}{a + b = c + 1 \ \wedge \ a - b \equiv 1 \mod 3}.
\]
\end{example}

\paragraph{Vector Addition Systems with states.}

A $d$-dimensional \emph{\vas with states} (\vass) is a triple $V = (s, T, Q)$,
where $Q$ is a finite set of \emph{states}, $s \in Q \times \N^d$ is the \emph{source} configuration
and $T \subseteq_\fin Q \times \Z^d \times Q$ is a finite set of \emph{transitions}.
Similarly as in case of \vases, a \emph{run} $\rho$ of a VASS $V = (s, T, Q)$ is a sequence
\[
(q_0, v_0, s_0, q_1, v_1), \ldots, (q_{n-1}, v_{n-1}, s_{n-1}, q_n, v_n)
\in Q \times \N^d \times \Z^d \times Q \times \N^d
\]
such that $(q_0, v_0) = s$ and for all $i \in \{0, \ldots, n-1\}$
we have $(q_i, s_i, q_{i+1}) \in T$  and $v_i + s_i = v_{i+1}$.
We write $\target{\rho} =(q_n, v_n)$.
The \emph{reachability set} of a \vass $V$ in state $q$ is
\[
\reach_q(V) \ = \ \setof{v \in \N^d}{(q, v) = \target{\rho} \text{ for some run } \rho}\enspace.
\]
The family of all such reachability sets of all \vasses we denote as $\reachvass$.
\begin{example}[cf.~\cite{DBLP:journals/tcs/HopcroftP79}] \label{ex:vassexponent}
	Let $V$ be a 3-dimensional \vass with two states, $p$ and $p'$, the source configuration $(p, (1,0,0))$, and four transitions:
	\begin{align*} \label{eq:3vas}
	(p, (-1,1,0), p), \quad (p, (0,0,0), p'), \quad (p', (2,-1,0), p'), \quad (p', (0,0,1), p). 
	\end{align*}
	Then $\reach_p(V) = \setof{(a,b,c) \in \N^3} {1 \leq a+b \leq 2^c}$.
\end{example}


\section{Sections}

\vasreach sets are central for this paper.
However, in order to make this family of sets more robust,
we prefer to consider the slightly larger family of \emph{sections} of \vasreach sets.
The intuition about a section is that we fix values on a subset of coordinates in vectors,
and collect all the values that can occur on the other coordinates.
For a vector $u \in \N^d$ and a subset $I \subseteq \{1, \ldots, d\}$ of coordinates,
by $\proj I u \in \N^{|I|}$ we denote the \emph{$I$-projection} of $u$, i.e., the vector obtained from $u$ by removing coordinates not belonging to $I$.
The projection extends element-wise to sets of vectors $S \subseteq \N^d$, denoted $\proj I S$.
For a set of vectors $S \subseteq \N^d$, a subset $I \subseteq \{1, \ldots, d\}$, and a vector $u \in \N^{d-|I|}$,
the \emph{section} of $S$ w.r.t. $I$ and $u$ is the set
\begin{align*}
	\sec {I, u} S := \proj I {\{ v \in S \mid \proj {\{1, \ldots, d\} \setminus I} v = u \}} \subseteq \N^{|I|}\enspace.
\end{align*}
We denote by $\secreachvas$ the family of all sections of \vasreach sets,
which we abbreviate as \emph{\vas{\sepsep}sections below}.
Similarly, the family of all sections of \vass-reachability sets we denote by $\secreachvass$.
\begin{example}\label{ex:sec}
Consider the \vas $V$ from Example~\ref{ex:vasreach}.
For $I = \{1, 2\}$ and $u = 7 \in \N^1$ we have
\[
\sec{I, u} {\reach(V)} \  = \ \{(0, 8), (3, 5), (6, 2)\}\enspace.
\]
\end{example}

Note that in a special case of $I = \{1, \ldots, d\}$, when $u$ is necessarily the empty vector, $\sec {I, u} S = S$.
Thus $\reach(\text{\vas})$ is a subfamily of $\secreach(\text{\vas})$, and likewise for {\vass}es.
We argue that \vas{\sepsep}sections are a more robust class than \vasreach sets.
Indeed, as shown below \vas{\sepsep}sections are closed under positive boolean combinations,
which is not the case for \vasreach sets.


Reachability sets of \vases are a strict subfamily of reachability sets of \vases with states,
which in turn are a strict subfamily of sections of reachability sets of \vases.
However, when sections of reachability set are compared, there is no difference between \vases and \vases with states,
which motivates considering sections in this paper.
These observations are summarized in the following propositions:
\begin{proposition}\label{prop:vas-expressivity}
$\reachvas \subsetneq \reachvass \subsetneq \secreachvas$.
\end{proposition}
\begin{proof}
	In order to prove strictness of the first inclusion,
	consider the \vass $V$ from Example~\ref{ex:vassexponent}.
	The reachability set $\reach_p(V)$ is not semilinear;
	on the other hand the reachability sets of of 3-dimensional \vases are always semilinear~\cite{DBLP:journals/tcs/HopcroftP79}.

	Now we turn to the second inclusion.
	It is folklore that for a $d$-dimensional \vass $V$ with $n$ states and $m$ transitions
	one can construct a $(d+n+m)$-dimensional \vas $V'$ simulating $V$.
	Among the new coordinates, $n$ correspond to states and $m$ to transitions.
	For a transition $t = (q, v, q')$ of $V$ there are two transitions in $V'$:
	the first one subtracts $1$ on the coordinate corresponding to state $q$
	and adds $1$ on the coordinate corresponding to $t$;
	the second one subtracts $1$ on the coordinate corresponding to $t$,
	adds $1$ on the coordinate corresponding to $q'$,
	and adds $v$ on the original $d$ coordinates.
	Finally, if $(q_0, v_0)$ is the initial configuration of $V$,
	then the initial configuration of $V'$ is a copy of $v_0$ on the original $d$ dimensions,
	equals $1$ on the coordinate corresponding to $q_0$,
	and equals $0$ on the rest of the new coordinates.
	Then the reachability set $\reach_q(V)$ equals the section of $\reach(V')$
	obtained by fixing the coordinate corresponding to $q$ to $1$
	and all other new coordinates to $0$.

	For strictness of the second inclusion, apply the above-mentioned transformation to the \vass $V$ 
	from Example~\ref{ex:vassexponent}, in order to obtain a 9-dimensional \vas $V'$. 
	The section of $\reach(V')$ that fixes the second original coordinate to $0$,
	the coordinate corresponding to state $p$ to $1$,
	and all the other new coordinates to $0$ is $S := \setof {(a,b) \in \N^2} {0 \leq a \leq 2^b}$.
	This 2-dimensional set is not semilinear, while reachability sets of 2-dimensional {\vass}es are 
	always semilinear~\cite{DBLP:journals/tcs/HopcroftP79}.
	Thus $S$ is not a 2-dimensional \vasreach set.  
\end{proof}

\begin{proposition}\label{prop:vas-expressivity}
$\secreachvas = \secreachvass$.
\end{proposition}
\begin{proof}
	One inclusion is obvious, since \vasses are more general than \vases,
	and the same holds when taking sections.
	For the other directions,
	consider a \vass $V$ and a section thereof $S := \sec {I,v}{\reach_q(V)}$.
	Reconsider the folklore construction of a \vas $V'$ that simulates $V$
	(cf. the proof of the previous Proposition~\ref{prop:vas-expressivity}).
	The section of the reachability set of $\reach(V')$
	that fixes the coordinate corresponding to $q$ to $1$,
	all the other new coordinates to $0$,
	and all the original coordinates not belonging to the set $I$ as in vector $v$,
	equals $S$. 
	\qed
\end{proof}

\begin{remark}
In the similar vein one shows that reachability sets of Petri nets include $\reachvas$ and are included in $\reachvass$.
Therefore, as long as sections are considered, there is no difference between \vases, Petri nets, and \vasses.
In consequence, our results apply not only to {\vas}es, but to all the three models.
\end{remark}

We conclude this section by proving a closure property of \vas sections.
\begin{proposition}\label{prop:union:intersection}
	The family of \vas sections is closed under positive boolean combinations.
\end{proposition}
\begin{proof}
	We only sketch the proof. 
	For closure under union, we just use nondeterminism to guess which \vas to run.
	Dealing with sections is straightforward since
	1) we can assume w.l.o.g.~that sections are done w.r.t. the 0 vector,
	2) by padding coordinates we can assume that the two input \vases have the same dimension,
	and 3) by reordering coordinates we can guarantee that the coordinates that are projected away appear all together on the right
	(the same simplifying assumptions will be made in Sections~\ref{sec:modular-vas-sep} and \ref{sec:unary-vas-sep};
	cf. the details just before Lemma~\ref{lem:modular:nonsep:witness}).
	For closure under intersection, we proceed under similar assumptions,
	and the intuition is to run the first \vas forward in two identical copies,
	and then to run backward the second \vas only in the second copy,
	using a section to make sure that the second \vas is accepting,
	and then project away the second copy.
	\qed
\end{proof}


\section{Results}
As our main technical contribution, we prove decidability of the modular\sepsep and unary{\sepsep}separability problems 
for the class of sections of \vasreach sets.

\begin{theorem}\label{thm:modular-vas-sep}
	The modular{\sepsep}separability problem for \vas{\sepsep}sections is decidable.
\end{theorem}

\begin{theorem}\label{thm:unary-vas-sep}
	The unary{\sepsep}separability problem for \vas{\sepsep}sections is decidable.
\end{theorem}
The proofs are postponed to Sections~\ref{sec:linear}--\ref{sec:unary-vas-sep}.
Furthermore, as a corollary of Theorem~\ref{thm:unary-vas-sep} we derive decidability of two commutative variants of 
the regular{\sepsep}separability of \vaslangs 
(formulated in~Theorems~\ref{cor:comreg-sep-vas} and~\ref{cor:reg-sep-pi-vas} below).

To consider languages instead of reachability sets, we need to assume that transitions of a \vas are \emph{labeled} 
by elements of an alphabet $\Sigma$, and thus every run is labeled by a word over $\Sigma$ obtained by concatenating labels of consecutive transitions of a run. We allow for silent transitions labeled by $\varepsilon$, i.e., transitions that do not contribute to the labeling of a run.
The language $L(V)$ of a \vas $V$ contains labels of those runs of $V$ that end in an \emph{accepting} configuration.
Our results work for several variants of acceptance; for instance, for a given fixed configuration $v_0$, 
\begin{itemize}
\item we may consider a configuration $v$ accepting if $v \geq v_0$ (this choice yields so called \emph{coverability languages}); or 
\item we may consider a configuration $v$ accepting if $v = v_0$ (this choice yields \emph{reachability languages}).
\end{itemize}

The Parikh image of a word $w\in \Sigma^*$, for a fixed total ordering $a_1 < \ldots < a_d$ of $\Sigma$,
is a vector in $\N^d$ whose $i$th coordinate 
stores the number of occurrences of $a_i$ in $w$. We lift the operation element-wise to languages,
thus the Parikh image of a language $L$, denoted $\parimage(L)$, is a subset of $\N^d$.
Two words $w, v$ over $\Sigma$ are \emph{commutative equivalent} if their Parikh images are equal.
The \emph{commutative closure} of a language $L \subseteq \Sigma^*$, denoted $\cc(L)$, is the language containing all words $w \in \Sigma^*$ commutative equivalent to some word $v\in L$.
A language $L$ is \emph{commutative} if it is invariant under commutative equivalence, i.e., $L = \cc(L)$. 
Unary sets of vectors are exactly the Parikh images of commutative regular languages; reciprocally, 
commutative regular languages are exactly the inverse Parikh images of unary sets.
Note that a commutative language is uniquely determined by its Parikh image.

As a corollary of Theorem~\ref{thm:unary-vas-sep} we deduce decidability of the following two commutative variants of 
the regular{\sepsep}separability of \vaslangs:
\begin{itemize}
\item
\emph{commutative regular{\sepsep}separability of \vaslangs}: given two \vases $V, V'$,
decide whether there is a commutative regular language $R$ that includes $L(V)$ and is disjoint from $L(V')$;
\item
\emph{regular{\sepsep}separability for commutative closures of \vaslangs}: given two \vases $V, V'$,
decide whether there is a regular language $R$ that includes $\cc(L(V))$ and is disjoint from $\cc(L(V'))$.
\end{itemize}

\begin{theorem}\label{cor:comreg-sep-vas}
Commutative regular{\sepsep}separability is decidable for \vaslangs.
\end{theorem}
Indeed, given two \vases $V, W$ one easy constructs another two \vases $V', W'$
s.t. $\parimage(L(V))$ is a section of $\reach(V')$, and similarly for $W'$.
By the tight correspondence between commutative regular languages and unary sets,
we observe that $L(V)$ and $L(W)$ are separated by a commutative regular language if, and only if,
their Parikh images $\parimage(L(V))$ and $\parimage(L(W))$ are separated by a unary set,
which is is decidable by Theorem~\ref{thm:unary-vas-sep}.
%
\begin{theorem}\label{cor:reg-sep-pi-vas}
Regular{\sepsep}separability is decidable for commutative closures of \vaslangs.
\end{theorem}
Similarly as above, we reduce to unary separability of \vas reachability sets (which is decidable once again by Theorem~\ref{thm:unary-vas-sep}),
which is immediate once one proves the following crucial observation.
\begin{lemma} \label{lem:com-sep}
	Two commutative languages $L, K \subseteq \Sigma^*$ are regular{\sepsep}separable if, and only if, their Parikh images are 
	unary{\sepsep}separable.
\end{lemma}
\begin{proof}
	We start with the ``if'' direction.
	Let $\parimage(K)$ and $\parimage(L)$ be separable by some unary set $U \subseteq \N^d$.
	Let $S = \setof { w\in \Sigma^* }{\parimage(w) \in U}$. 
	It is easy to see that $S$ is (commutative) regular since $U$ is unary, and that $S$ separates $K$ and $L$.  

	Now we turn to the ``only if'' direction.
	Let $K$ and $L$ be separable by a regular language $S$, say $K \subseteq S$ and $S \cap L = \emptyset$.
	Let $M$ be the syntactic monoid of $S$ and $\omega$ be its idempotent power,
	i.e., a number such that for every $m \in M$ it holds $m^\omega = m^{2 \omega}$.
	In particular, for every word $u \in \Sigma^*$ we have
	\begin{equation}\label{eq:omega-power}
		u v^\omega w \in L \iff u v^{2\omega} w \in S\enspace;
	\end{equation}
	in other words, one can substitute $v^\omega$ by $v^{2\omega}$ and vice versa in every context.
	Let $\Sigma = \{a_1, \ldots, a_d\}$.
	For $u = (u_1, \ldots, u_d) \in \N^d$ define a word $w_u = a_1^{u_1} \cdots a_d^{u_d}$.
	For every $u, v \in \N^d$ such that $u \cong_\omega v$,
	by repetitive application of~\eqref{eq:omega-power} we get
	$w_u \in S$ iff $w_v \in S$. 
	As $K$ is commutative and $K \subseteq S$, we have $w_u \in S$ for all $u \in \parimage(K)$;
	similarly, we have $w_v \not\in S$ for all $v \in \parimage(L)$.
	Therefore for all $u \in \parimage(K)$, $v \in \parimage(L)$ we have $u \not\cong_\omega v$.
	Let $U = \setof{x \in \N^d}{\prettyexists{y \in \parimage(K)}{x \cong_\omega y}}$. 
	The set $U$ separates $\parimage(K)$ and $\parimage(L)$ and,
	being a union of $\cong_\omega$ equivalence classes, it is unary.
	\qed
\end{proof}


\section{Modular\sepsep and unary{\sepsep}separability of linear sets}
\label{sec:linear}

The rest of the paper is devoted to the proofs of Theorems~\ref{thm:modular-vas-sep} and~\ref{thm:unary-vas-sep}.
In this section we prove that modular{\sepsep}separability of linear sets is decidable%
\footnote{While decidability follows from~\cite{DBLP:journals/ipl/ChoffrutG06} and is thus not a new result,
we provide here another simple proof to make the paper self-contained.},
and provide a condition on linear sets that makes modular{\sepsep}separability equivalent to unary{\sepsep}separability.
The two results, stated in Lemmas~\ref{cor:modular:separability:linear} and~\ref{lem:F-linked} below, 
respectively, are used in Sections~\ref{sec:modular-vas-sep} and~\ref{sec:unary-vas-sep}, 
where the proofs of Theorems~\ref{thm:modular-vas-sep} and~\ref{thm:unary-vas-sep} are completed.

\paragraph{Linear combinations modulo $n$.}
We start with some preliminary results from linear algebra.
For $n \in \N$,
let $\lin_n^{\geq 0}(v_1, \ldots, v_k)$ be the closure of $\lin^{\geq 0}(v_1, \ldots, v_k)$ modulo $n$, i.e.,
\begin{align*}
	\lin_n^{\geq 0}(v_1, \ldots, v_k) = \setof { v \in \N^d} {\prettyexists{u \in \lin^{\geq 0}(v_1, \ldots, v_k)}{v \equiv_n u}}\enspace.
\end{align*}
Similarly one defines $\lin_n(v_1, \ldots, v_k)$ be the closure of $\lin(v_1, \ldots, v_k)$ modulo $n$.
Observe however that $\lin_n(v_1, \ldots, v_k) = \lin_n^{\geq 0}(v_1, \ldots, v_k)$.
Indeed, if $v \equiv_n l_1 v_1 + \ldots + l_k v_k$ for $l_1, \ldots, l_k \in \Z$ then 
$v \equiv_n (l_1 + N n) v_1 + \ldots + (l_k + N n) v_k$ for any $N\in\N$.

\begin{lemma}
	\label{lem:lin:comb:mod}
	\label{lem:solution-modulo}
	$\lin (v_1, \dots, v_k) \ = \ \bigcap_{n>0} \lin_n^{\geq 0} (v_1, \dots, v_k)$.
\end{lemma}

\begin{proof}
	The left-to-right inclusion is immediate: 
	for any $n \in \N$ we have
	\[ \lin (v_1, \dots, v_k) \ \subseteq \ \lin_n (v_1, \dots, v_k)  \ = \  \lin_n^{\geq 0} (v_1, \dots, v_k)\enspace. \]
	

	For the right-to-left inclusion 
	we take an algebraic perspective,
	and treat $S := \lin(v_1, \ldots, v_k)$ as a subgroup of $\Z^d$ generated by $F = \{v_1, \ldots, v_k\}$. 
	Let $I$ be the set of all $d$ unit vectors in $\Z^d$.
	For every $n\in \Nplus$, let $n \Z^d$ denote the subgroup of $\Z^d$ generated by $n I$, and let $S_n$ be the subgroup of $\Z^d$ generated by $F \cup (n I)$. 
	In algebraic terms, our obligation is to show that 
	\begin{align} \label{eq:toprove}
	\bigcap_{n\in \Nplus} S_n \ \subseteq \ S\enspace.
	\end{align}
	Let $G := \Z^d/S$ be the quotient group and consider the quotient group homomorphism $h : \Z^d \to G$.
	It is legal, as every subgroup of an abelian group is normal, thus we can consider a quotient with respect to it.
	We have thus $\ker(h) = \setof{x\in \Z^d}{h(x) = \zeroel{G}} = S$, where $\zeroel{G}$ is the zero element of $G$.
	Now~\eqref{eq:toprove} is equivalent to
	\begin{align*}
	h\big(\bigcap_{n\in \Nplus} S_n \big) \ = \ \{\zeroel{G}\}\enspace,
	\end{align*}
	which will immediately follow, once we manage to show
	\begin{align*}
	\bigcap_{n\in \Nplus} h(S_n) \ = \ \{\zeroel{G}\}\enspace.
	\end{align*}
	Observe that $h(S_n) = h(n \Z^d)$, for every $n\in \Nplus$, and hence we may equally well demonstrate:
	\begin{align} \label{eq:toprovereally}
	\bigcap_{n\in \Nplus} h(n \Z^d) \ = \ \{\zeroel{G}\}\enspace.
	\end{align}
	The group $G$, being a finitely generated abelian group, is isomorphic to the direct
	product of a finite group $G_1$ (let $l$ be its order, i.e., the number of its elements)
	and $G_2 = \Z^k$, for some $k\in \N$
	(see for instance Theorem 2.2, p.~76, in~\cite{Hungerford74}).
	For showing~\eqref{eq:toprovereally}, consider an element $g\in G$ which belongs to $h(n \Z^d)$ for all $n\in \Nplus$,
	and its two projections	$g_1$ and $g_2$ in $G_1$ and $G_2$, respectively.
	As $g \in h(l \Z^d)$, then necessarily $g_1 = l\cdot g'$ for some $g' \in G_1$,
	and since the order of every element divides the order of the group $l$,
	we have $g_1 = \zeroel{G_1}$.
	Similarly, we deduce that $g_2 = \zeroel{G_2}$; indeed, this is implied by the fact that for every $n\in \Nplus$, $g_2 = n g'$ for some $g' \in G_2$.
	Thus $g = \zeroel{G}$ as required.
	\qed
\end{proof}

\paragraph{Modular{\sepsep}separability.}

In the rest of the paper, we heavily rely on the following straightforward characterization of modular{\sepsep}separability:
\begin{proposition}\label{prop:modular-inseparable-pairs}
	Two sets $U, V \subseteq \N^d$ are modular{\sepsep}separable if, and only if, there exists $n \in \N$ such that
	for all $u \in U$, $v \in V$ we have $u \not\equiv_n v$.
\end{proposition}

\begin{proof}
	If $U, V$ are separable by some $n$-modular set,
	then for all $u \in U, v \in V$ we have $u \not\equiv_n v$.
	On the other hand, if for all $u \in U, v \in V$ we have $u \not\equiv_n v$,
	then the modular set $S = \setof{s \in \N^d}{\prettyexists{u \in U}{s \equiv_n u}}$
	separates $U$ and $V$.
	\qed
\end{proof}

\begin{lemma}
	\label{lem:modular:separability:linear}
	Two linear sets $\set b + \lin^{\geq 0}(P)$ and $\set c + \lin^{\geq 0}(Q)$
	are not modular{\sepsep}separable if, and only if,
	$b - c \in \lin(P \cup Q)$.
\end{lemma}

\begin{proof}
	Let $L = \set b + \lin^{\geq 0}(P)$ and $M = \set c + \lin^{\geq 0}(Q)$,
	with $P = \set{p_1, \ldots, p_m}$ and $Q = \set{q_1, \ldots, q_n}$.
	First we show the ``if'' direction.
	By Proposition~\ref{prop:modular-inseparable-pairs},
	it is enough to show that, for every $n \in \N$,
	there exist two vectors $u \in L$ and $v \in M$ s.t. $u \equiv_n v$.
	Fix an $n \in \N$.
	By assumption, 
	we have $b - c \in \lin(P \cup Q)$, and thus $c - b \in \lin(P \cup Q) = \lin(P \cup -Q)$.
	By Lemma~\ref{lem:lin:comb:mod},
	$c - b \in \lin_n^{\geq 0}(P \cup -Q)$,
	i.e., there exist
	$\delta \in \lin^{\geq 0}(P)$ and $\gamma \in \lin^{\geq 0}(Q)$ such that
	%
	$c - b \equiv_n \delta - \gamma$.
	%
	Thus, if we take $u = b + \delta$ and $v = c + \gamma$
	we clearly have $u - v = (b - c) + (\delta - \gamma) 
	\equiv_n (b - c) + (c - b) = 0$,
	and thus $u \equiv_n v$.

	For the ``only if'' direction,
	assume that $L$ and $M$ as above are not modular{\sepsep}separable.
	By Proposition~\ref{prop:modular-inseparable-pairs},
	for every $n \geq 0$ there exist vectors $u_n \in L$ and $v_n \in M$ s.t. $u_n \equiv_n v_n$.
	By definition,
	$u_n = b + \delta_n$ and $v_n = c + \gamma_n$,
	for some $\delta_n \in \lin^{\geq 0}(P)$ and $\gamma_n \in \lin^{\geq 0}(Q)$.
	Since $u_n \equiv_n v_n$, we have $b - c \equiv_n \gamma_n - \delta_n \in \lin(P \cup Q)$,
	%
	and thus $b - c \in \lin_n^{\geq 0}(P \cup Q)$.
	Since $n$ was arbitrary, by Lemma~\ref{lem:solution-modulo}
	we have $b - c \in \lin(P \cup Q)$,	as required. 
	\qed
\end{proof}
Since the condition in the lemma above is effectively testable being an instance of solvability of systems of linear Diophantine equations,
we get the following corollary:   

\begin{corollary}
	\label{cor:modular:separability:linear}
	Modular{\sepsep}separability of linear sets is decidable. 
\end{corollary}

\begin{remark}
	Since linear Diophantine equations are solvable in polynomial time,
	we obtain the same complexity for modular{\sepsep}separability of linear sets.
	This observation however will not be useful in the sequel.
\end{remark}

\paragraph{Unary{\sepsep}separability.}

We start with a characterization of unary{\sepsep}separability,
which is the same as Proposition~\ref{prop:modular-inseparable-pairs},
with unary equivalence $\cong_n$ in place of modular equivalence $\equiv_n$.
(Recall that unary equivalence is modular equivalence ``above a threshold'',
i.e., $u \cong_n v$ holds for two vectors $u, v \in \N^d$ if, for every component $1 \leq i \leq d$,
either $u[i] = v[i] \leq n$, or $u[i], v[i] \geq n$ and $u[i] \equiv_n v[i]$.)
\begin{proposition}\label{prop:inseparable-pairs}
	Two sets $U, V \subseteq \N^d$ are unary{\sepsep}separable if, and only if,
	there exists $n \in \N$ such that, for all $u \in U$ and $v \in V$, we have $u \not\cong_n v$.
\end{proposition}
%

%
We say that a set of vectors $U \subseteq \N^d$ is \emph{diagonal} if,
for every threshold $x \in \N$,
there exists a vector $u \in U$ which is strictly larger than $x$ in every component.
Let $I \subseteq \set{1, \dots, d}$ be a set of coordinates.
Two set of vectors $U, V \subseteq \N^d$ are \emph{$I$-linked}
if there exists a sectioning vector $u \in \N^{d-\card I}$
s.t. $\proj {\set {1, \dots, d} \setminus I} U = \proj {\set {1, \dots, d} \setminus I} V = \set{u}$
and $\proj I U$, $\proj I V$ are diagonal.
The sets $U, V$ are \emph{linked} if they are $I$-linked for some $I \subseteq \set{1, \dots, d}$.
\begin{lemma}
	\label{lem:F-linked}
	Let $U, V \subseteq \N^d$ be two linked linear sets. Then,
	$U$ and $V$ are unary{\sepsep}separable if, and only if, they are modular{\sepsep}separable.
	%
		%
		%
		%
		%
	%
\end{lemma}
\begin{proof}
	Let $U$ and $V$ be two linked linear sets.
	One direction is obvious since modular{\sepsep}separability implies unary{\sepsep}separability.
	For the other direction, let $U$ and $V$ be modular{\sepsep}nonseparable,
	and we show that they are unary{\sepsep}nonseparable either.
	By Lemma~\ref{prop:modular-inseparable-pairs},
	there exists a sequence of vectors $u_n \in U$ and $v_n \in V$ s.t. $u_n \equiv_n v_n$.
	We construct a new sequence $u_n' \in U$ and $v_n' \in V$ s.t. $u_n' \cong_n v_n'$,
	which will then show that $U$ and $V$ are not unary{\sepsep}separable by Lemma~\ref{prop:inseparable-pairs}.
	Since $U$ and $V$ are linked,
	there exist a set of coordinates $I \subseteq \set{1, \dots, d}$
	and a sectioning vector for the remaining coordinates $u \in \N^{d - \card I}$
	s.t. 1) $\proj {\set {1, \dots, d} \setminus I} U = \proj {\set {1, \dots, d} \setminus I} V = \set{u}$
	and 2) $\proj I U, \proj I V$ are diagonal.
	In particular, by 1) the two sequences $u_n$ and $v_n$ project to $u$ on the complement of $I$,
	i.e., $\proj {\set {1, \dots, d} \setminus I} {u_n} = \proj {\set {1, \dots, d} \setminus I} {v_n} = \set u$.
	Moreover, for any $n \in \N$,
	since $\proj I {u_n} \in \proj I U$, and the latter set is diagonal by 2),
	there exists an increment $\delta_n \in \N^{\card I}$ s.t. $\proj I {u_n} \leq \proj I {u_n} + \delta_n \in \proj I U$.
	Moreover, since $U$ is a linear set, $\delta_n$ can be chosen to have its components multiple of $n$.
	Let $u_n'$ be $\proj I {u_n} + \delta_n$ on coordinates $I$, and $u$ on the remaining ones.
	By the choice of $\delta_n$, $u_n' \equiv_n u_n$, and, moreover,
	$u_n'$ is larger than $n$ on coordinates $I$.
	The vector $v_n'$ can be constructed similarly from $v_n$.
	We thus have $u_n' \cong_n v_n'$,
	since on coordinates $I$ both $u_n'$ and $v_n'$ are above $n$,
	and on the remaining coordinates they are equal to $u$.
	\qed
\end{proof}

\begin{remark}
	The unary{\sepsep}separability problem is decidable for linear sets, as shown in~\cite{DBLP:journals/ipl/ChoffrutG06}, 
	but we will not need this fact in the sequel.
	Moreover, it will follow from our stronger decidability result about the more general \vas reachability sets stated in Theorem~\ref{thm:unary-vas-sep}
	(since linear sets are included in \vas reachability sets).
\end{remark}

\section{Modular{\sepsep}separability of \vas{\sepsep}sections} 
\label{sec:modular-vas-sep}

In this section we prove Theorem~\ref{thm:modular-vas-sep}, 
and thus provide an algorithm to decide modular{\sepsep}separability for \vasreach sets.
Given two \vas{\sepsep}sections $U$ and $V$,
the algorithm runs in parallel two semi-decision procedures: one (positive) which looks for a witness of separability,
and another one (negative) which looks for a witness of nonseparability.
Directly from the characterization of Proposition~\ref{prop:modular-inseparable-pairs},
the positive semi-decision procedure simply enumerates all candidate moduli $n\in \N$ and checks whether 
$u \not\equiv_n v$ for all $u \in U$ and $v \in V$.
The latter condition can be decided by reduction to the \vas (non)reachability problem \cite{DBLP:conf/stoc/Mayr81,DBLP:conf/lics/LerouxS15}.
\begin{lemma}\label{lem:modular-positive-witness}
	For two VAS{\sepsep}sections $U$ and $V$ and a modulus $n \in \N$,
	it is decidable whether there exist $u \in U$ and $v \in V$ s.t. $u \equiv_n v$.
\end{lemma}

\begin{proof}
	Recall that $U$ is obtained from the reachability set of a \vas by fixing values $\bar u$ on some coordinates, 
	and projecting to the remaining coordinates; and likewise $V$ is obtained, by fixing values $\bar v$ on some coordinates. 
	We modify the two \vases by allowing each non-fixed coordinate to be decremented by $n$,
	and we check whether the two thus modified \vases admit a pair of reachable vectors $u, v$
	that agree on fixed coordinates with $\bar u$ and $\bar v$, respectively,
	and on all the non-fixed coordinates are equal and smaller than $n$.
	\qed
\end{proof}

It remains to design the negative semi-decision procedure, which is the nontrivial part.
In Lemma~\ref{lem:modular:nonsep:witness}, we show that if two \vasreach sets are not modular{\sepsep}separable,
then in fact they already contain two \emph{linear} subsets which are not modular{\sepsep}separable.
In order to construct such linear witnesses of nonseparablity,
we use the theory of well quasi orders and some elementary results in algebra,
which we present next.

\paragraph{The order on runs.}
A quasi order $(X, \preccurlyeq)$ is a \emph{well quasi order} (wqo) if for every infinite sequence $x_0, x_1, \ldots \in X$
there exist indices $i, j \in \N, i < j$, such that $x_i \preccurlyeq x_j$.
It is folklore that if $(X, \preccurlyeq)$ is a wqo, then in every infinite sequence
$x_0, x_1, \ldots \in X$ there even exists an infinite monotonically non-decreasing subsequence $x_{i_1} \preccurlyeq x_{i_2} \preccurlyeq \ldots$.
We will use Dickson's and Higman's Lemmas to define new wqo's on pairs and sequences.
For two quasi orders $(X, \leq_X)$ and $(Y, \leq_Y)$,
let the product $(X \times Y, \leq_{X \times Y})$ be ordered componentwise by $(x, y) \leq_{X \times Y} (x', y')$ if $x \leq_X x'$ and $y \leq_Y y'$.
By Dickson's Lemma~\cite{Dickson}, if both $(X, \leq_X)$ and $(Y, \leq_Y)$ are wqos, then $(X \times Y, \leq_{X \times Y})$ is a wqo too.
As a corollary of Dickson's Lemma, if two quasi orders $(X, \leq_1)$ and $(X, \leq_2)$ on the same domain are wqos, then the quasi order defined as the conjunction of $\leq_1$ and $\leq_2$ is a wqo too.
For a quasi order $(X, \leq)$,
let $(X^*, \leq_*)$ be quasi ordered by the subsequence order $\leq_*$,
defined as $x_1 x_2 \cdots x_k \leq_* y_1 y_2 \ldots y_m$ if there exist $1 \leq i_1 < \ldots < i_k \leq m$ such that
$x_j \leq y_{i_j}$ for all $j \in \{1, \ldots, k\}$.
By Higman's Lemma~\cite{Higman}, if $(X, \leq)$ is a wqo then $(X^*, \leq_*)$ is a wqo too.

By considering the finite set of transitions $T$ well quasi ordered by equality,
we define the order $\leq^1$ on triples $\N^d \times T \times \N^d$ componentwise as
$(u, s, u') \leq^1 (v, t, v')$ if $u \leq v$, $s = t$, and $u' \leq v'$,
which is a wqo by Dickson's Lemma.
%
%
We further extend $\leq^1$ to an order $\runord$ on runs
by defining, for two runs $\rho$ and $\sigma$ in $(\N^d \times T \times \N^d)^*$,
$\rho \runord \sigma$ if $\rho \leq^1_* \sigma$ and $\target{\rho} \leq \target{\sigma}$.%
\footnote{A weaker version of this order not considering target configurations was defined in~\cite{DBLP:journals/tcs/Jancar90}.}
Here, $\leq^1_*$ is the extension of $\leq^1$ to sequences, and thus a wqo by Higman's Lemma,
which implies that $\runord$ is itself a wqo by the corollary of Dickson's Lemma.

\begin{proposition}
	\label{prop:run-wqo}
	$\runord$ is a well quasi order.
\end{proposition}

\begin{lemma}\label{lem:additive}
	Let $\rho$, $\rho_1$, and $\rho_2$ be runs of a \vas s.t. $\rho \runord \rho_1, \rho_2$.
	There exists a run $\rho'$ s.t. $\rho \runord \rho'$
	and $\target{\rho'} - \target{\rho} = (\target{\rho_1} - \target{\rho}) + (\target{\rho_2} - \target\rho)$.
\end{lemma}

\begin{proof}
	The proof is almost identical to the proof of Proposition 5.1. in~\cite{DBLP:journals/corr/LerouxS15}.
	Let the \vas be $(s, T)$, and let $\rho = v_0 \trans{t_0} v_1 \trans{t_1} \cdots \trans{t_{n-1}} v_n$, where $v_0 = s$.
	Then $\rho_i$, for $i \in \{1, 2\}$ is of the form
	\begin{eqnarray*}
		\rho_i = & v_0 \trans{\rho^i_0} v_0 + \delta^i_0 \trans{t_0} v_1 + \delta^i_0 \trans{\rho^i_1} v_1 + \delta^i_1
		\trans{t_1} v_1 + \delta^i_2 \trans{\rho^i_2} \cdots \\
		& \trans{\rho^i_{n-1}} v_{n-1} + \delta^i_{n-1} \trans{t_{n-1}} v_n + \delta^i_{n-1} \trans{\rho^i_n} v_n + \delta^i_n\enspace,
	\end{eqnarray*}
	where for all $i \in \{1, 2\}$ and $j \in \{0, \ldots, n\}$ we have $\delta^i_j \geq 0$.
	Thus by letting $\rho' := \rho^1_0 \rho^2_0 t_0 \rho^1_1 \rho^2_1 t_1 \rho^1_2 \rho^2_2
	\cdots \rho^1_{n-1} \rho^2_{n-1} t_{n-1} \rho^1_n \rho^2_n$
	we clearly have a run $v_0 \trans{\rho'} v_n + \delta^1_n + \delta^2_n$
	which indeed looks like
	\begin{eqnarray*}
		v_0 \trans{\rho^1_0} v_0 + \delta^1_0 \trans{\rho^2_0} v_0 + \delta^1_0 + \delta^2_0 \trans{t_0} v_1 + \delta^1_0 + \delta^2_0 \\
		\trans{\rho^1_1} v_1 + \delta^1_1 + \delta^2_0 \trans{\rho^2_1} v_1 + \delta^1_1 + \delta^2_1 \trans{t_1}
		v_2 + \delta^1_1 + \delta^2_1 \\
		\trans{\rho^1_2} \cdots \trans{t_{n-1}} v_n + \delta^1_{n-1} + \delta^2_{n-1} \\
		\trans{\rho^1_n} v_n + \delta^1_n + \delta^2_{n-1} \trans{\rho^2_n} v_n + \delta^1_n + \delta^2_n.
	\end{eqnarray*}
	This finishes the proof of Lemma~\ref{lem:additive}.
	\qed
\end{proof}
We formulate an immediate but useful corollary:

\begin{corollary}\label{cor:summing-runs}
	Let $\rho_0, \rho_1, \ldots, \rho_k$ be runs of a VAS
	s.t., for all $i \in \{1, \ldots, k\}$, $\rho_0 \runord \rho_i$,
	and let $\delta_i := \target{\rho_i} - \target{\rho_0} \geq 0$.
	For any $\delta \in \lin^{\geq 0}(\delta_1, \dots, \delta_k)$,
	there exists a run $\rho$ s.t. $\rho_0 \runord \rho$ and $\delta = \target{\rho} - \target{\rho_0}$.
\end{corollary}

We conclude this part by showing that any (possibly infinite) subset of $\Z^d$
can be overapproximated by taking linear combinations of a \emph{finite} subset thereof.
This will be important below in order to construct linear sets as witnesses of nonseparability.
\begin{lemma}\label{lem:finite-base}
	For every (possibly infinite) set of vectors $S \subseteq \Z^d$,
	there exist finitely many vectors $v_1, \dots, v_k \in S$
	s.t. $S \subseteq \lin(v_1, \dots, v_k)$.
\end{lemma}
\begin{proof}
	Treat $\Z^d$ as a freely finitely generated abelian group,
	and consider the subgroup $\lin(S)$ of $\Z^d$ generated by $S$, 
	i.e., the subgroup containing all linear combinations of finitely many elements of $S$.
	We use the following result in algebra: 
	every subgroup of a finitely generated abelian group is finitely generated (see for instance Corollary 1.7, p.~74, in~\cite{Hungerford74}).
	By this result applied to $\lin(S)$ we get a finite set of generators $F \subseteq \lin(S)$ s.t. $\lin(F) = \lin(S)$.
	Every element of $F$ is a linear combination of finitely many elements of $S$.
	Thus let $v_1, \ldots, v_k$	be all the elements of $S$ appearing as a linear combination of some element from $F$.
	Then clearly $F \subseteq \lin(v_1, \dots, v_k)$,
	and thus $S \subseteq \lin(S) = \lin(F) \subseteq \lin(\lin(v_1, \dots, v_k)) = \lin(v_1, \dots, v_k)$, as required.
	\qed
\end{proof}

\begin{remark}
	In fact one can show that the generating set $F$ has at most $d$ elements. However, no upper bound on $k$ follows, 
	and even for $d=1$ the number of vectors $k$ can be arbitrarily large.
	Indeed, let $p_1, \ldots, p_k$ be different prime numbers, let $u_i = (p_1 \cdot \ldots \cdot p_k) / p_i$
	and $S = \{u_1, \ldots, u_k\}$.
	Then for every $i \in \{1, \ldots, k\}$, the number $u_i$ is not a linear combination of numbers $u_j$, $j \neq i$, 
	as $u_i$ is not divisible by $p_i$, while all the others are. Therefore we need all the elements of $S$ in the set $\{v_1, \ldots, v_k\}$.
\end{remark}

\paragraph{Modular{\sepsep}nonseparability witness.}
We now concentrate on the negative semi-decision procedure.
Let $U, V \subseteq \N^d$ be two \vas{\sepsep}sections:
	\[U = \sec {I, \bar u} {R_U} \subseteq \N^d \quad \textrm{ and } \quad V = \sec {J, \bar v} {R_V} \subseteq \N^d,\]
	where $R_U \subseteq \N^{d_U}$ and $R_V \subseteq \N^{d_V}$ are the reachability sets
	of the two \vases $W_U$ and $W_V$, 
	and $I \subseteq \{1, \ldots, d_U\}$ and $J \subseteq \{1, \ldots, d_V\}$
	with $|I| = |J| = d$ are projecting coordinates,
	and $\bar u \in \N^{d_U - d}, \bar v \in \N^{d_V - d}$ are two sectioning vectors.

	Observe that by padding coordinates we can assume
	w.l.o.g.~that the two input \vases have the same dimension $d' = d_U = d_V$.
	Furthermore, we can also assume w.l.o.g.~that $\bar u = \bar v = 0$.
	Indeed, one can add an additional coordinate, such that for performing any transition it is necessary that
	this coordinate is nonzero and a special, final transition, which causes the additional coordinate to be equal zero
	and subtracts $\bar u$ (or $\bar v$) from the other coordinates. The result of adding this gadget is that now
	we can assume $\bar u = \bar v = 0$, but the section itself does not change.
	
	Finally, by reordering coordinates we can guarantee that the coordinates that are projected away appear on the same positions
	in both \vases, i.e., $I = J$.
	With these assumptions, we observe that modular{\sepsep}separability of sets $U, V \subseteq \N^d$ is equivalent to
	modular{\sepsep}separability of sets $U', V' \subseteq \N^{d'}$, defined as $U, V$ but
	\emph{without} projecting onto the subset $I$ of coordinates:
	\[
		U' = \{ v \in R_U \mid \proj {\{1, \ldots, d'\} \setminus I} v = 0 \} \qquad
		V' = \{ v \in R_V \mid \proj {\{1, \ldots, d'\} \setminus I} v = 0 \}\enspace.
	\]
We call the set $U'$ (resp.~$V'$) the \emph{expansion} of $U$ (resp.~$V$).	

We say that a linear set $L = \{b\} + \lin^{\geq 0}(p_1, \ldots, p_k) \subseteq \N^{d'}$
is a \emph{$U$-witness} if $W_U$ admits runs $\rho, \rho_1, \ldots, \rho_k$ such that 
\begin{align} \label{eq:witnesscond}
\begin{aligned}
b = \target{\rho} \in U' \\
b+p_i = \target{\rho_i} \in U' & \ \  \text{ for } i\in\{1,\ldots, k\} \\
 \rho \runord \rho_i  & \ \ \text{ for } i \in \{1, \ldots, k\}\enspace.
\end{aligned}
\end{align}
Analogously one defines $V$-witnesses, but with respect to $W_V$.

%

\newcommand{\keylemma}[2]{
\begin{lemma}
	\label{#2}
	For two \vas{\sepsep}sections $U, V \subseteq \N^d$, the following conditions are equivalent:
	\begin{enumerate}
	\item $U, V$ are not #1{\sepsep}separable;
	\item the expansions $U', V'$ of $U, V$ are not #1{\sepsep}separable;
	\item there exist linear subsets $L\subseteq U'$, $M \subseteq V'$ that are not #1{\sepsep}separable;
	\item there exist a $U$-witness $L$ and a $V$-witness $M$ that are not #1{\sepsep}separable.
	\end{enumerate}
\end{lemma}
}

\keylemma{modular}{lem:modular:nonsep:witness}
\begin{proof}
	Equivalence of points 1 and 2 follows by the definition of expansion. Point 4 implies 3, as a $U$-witness is necessarily
	a subset of the expansion $U'$ by Corollary~\ref{cor:summing-runs}.
	Point 3 implies 2,   
	since if two sets are separable, also subsets thereof are separable (moreover, the separator remains the same).
	It remains to show that 2 implies 4.    

	Let $U', V' \subseteq \N^{d'}$ be the expansions of two \vas{\sepsep}sections $U, V\subseteq \N^d$, as above,
	and assume that they are not modular{\sepsep}separable.
	We construct two linear sets $L, M \subseteq \N^{d'}$ constituting a $U$-witness and a $V$-witness, respectively.
	%
	By Proposition~\ref{prop:modular-inseparable-pairs},
	there exists an infinite sequence of pairs of reachable configurations
	$(u_0, v_0), (u_1, v_1), \ldots \in U' \times V'$ s.t. $u_n \equiv_n v_n$ for all $n \in \N$.
	By taking an appropriate infinite subsequence we can ensure that even $u_n \equiv_{n!} v_n$ for all $n \in \N$.
	Let us fix for every $n \in \N$ runs $\rho_n$ and $\sigma_n$ such that
		$u_n = \target{\rho_n}$ and $v_n = \target{\sigma_n}$.
	%
	Since $\runord$ is a wqo by Proposition~\ref{prop:run-wqo},
	we can extract a monotone non-decreasing subsequence,
	and thus we can ensure that even $\rho_0 \runord \rho_1 \runord \cdots$ and $\sigma_0 \runord \sigma_1 \runord \cdots$.
	Here we use the fact that $u_n \equiv_{n!} v_n$ in the original sequence,
	and thus $u_n \equiv_i v_n$ for every $i \in \{1, \dots, n\}$,
	consequently the new subsequence still has $u_n \equiv_n v_n$ for all $n\in\N$.
	%
	For all $n \in \N$, let $\delta_n := u_n - u_0$ and $\gamma_n := v_n - v_0$,
	and
	%
	%
	consider the set of corresponding differences $S_\textrm{inf} := \setof {\delta_n - \gamma_n} {n \in \N}$. 
	By Lemma~\ref{lem:finite-base},
	there exists a finite subset thereof $S := \set{\delta_{i_1} - \gamma_{i_1}, \dots, \delta_{i_k} - \gamma_{i_k}}$ 
	such that $S_\textrm{inf} \subseteq \lin(S)$,
	and thus there exist two finite subsets
	$P := \set{\delta_{i_1}, \dots, \delta_{i_k}}$  
	and
	$Q := \set{\gamma_{i_1}, \dots, \gamma_{i_k}}$  
	such that
	\begin{align}
		\label{eq:lin:subset}
		S_{\textrm{inf}} \subseteq \lin(P - Q) \subseteq \lin(P) - \lin(Q) \subseteq \lin^{\geq 0}_n(P) - \lin^{\geq 0}_n(Q)\enspace,
	\end{align}
	where the last inclusion follows from Lemma~\ref{lem:lin:comb:mod}.
	%
	%
	Let the two linear sets $L$ and $M$ be defined as
	\begin{align*}
		L &:= \set {u_0} + \lin^{\geq 0} (P)\enspace, \textrm { and } \\
		M &:= \set {v_0} + \lin^{\geq 0} (Q)\enspace.
	\end{align*}
	By the construction, $L$ is a $U$-witness and $M$ a $V$-witness. 
	It thus only remains to show that  $L$ and $M$ are not modular{\sepsep}separable.
	%
	%
	For any $n$,
	by Eq.~\ref{eq:lin:subset} we have $\delta_n - \gamma_n \equiv_n \delta_n' - \gamma_n'$
	for some $\delta_n' \in \lin^{\geq 0}(P)$ and $\gamma_n' \in \lin^{\geq 0}(Q)$.
	Consider now the two new infinite sequences $u_1',u_2',\dots \in L$ and $v_1',v_2', \dots \in M$
	defined, for every $n \geq 1$, as $u_n' := u_0 + \delta_n'$ and $v_n' := v_0 + \gamma_n'$.
	Then,
	\begin{align*}
		u_n' - v_n'	&=			(u_0 + \delta_n') - (v_0 + \gamma_n')  \\
		 			&=			(u_0 - v_0) + (\delta_n' - \gamma_n') 	&&\textrm{ (by def. of $\delta_n', \gamma_n'$) } \\
					&\equiv_n	(u_0 - v_0) + (\delta_n - \gamma_n)  \\
					&=			(u_0 + \delta_n) - (v_0 + \gamma_n)  \\
					&=			u_n - v_n \equiv_n 0							&&\textrm{ (by def. of $u_n, v_n$) },
	 \end{align*}
	and thus $u_n' \equiv_n v_n'$.
	This, thanks to the characterization of Proposition~\ref{prop:modular-inseparable-pairs},
	implies that $L$ and $M$ are not modular{\sepsep}separable.
	%
%
	%
	%
	%
	%
	\qed
\end{proof}

\begin{remark}
	Note that a modular{\sepsep}nonseparability witness exists even in the case when the two reachability sets $U, V$ have nonempty intersection. 
	In this case, it is enough to consider two runs $\rho_0$ and $\sigma_0$
	ending up in the same configuration $\target{\rho_0} = \target{\sigma_0}$,
	and considering the linear sets $L := M := \set {\target{\rho_0}}$.
\end{remark}

Using the characterization of Lemma~\ref{lem:modular:nonsep:witness}, the negative semi-decision procedure enumerates
all pairs $L, M$, where $L$ is a $U$-witness and $M$ is a $V$-witness and checks
whether $L$ and $M$ are modular{\sepsep}separable, which is decidable due to Lemma~\ref{lem:modular:separability:linear}.
Note that enumerating $U$-witnesses (and $V$-witnesses) amounts of enumerating
finite sets of runs $\{\rho, \rho_1, \ldots, \rho_k\}$ satisfying~\eqref{eq:witnesscond}.

\begin{remark}
	It is also possible to design another negative semi-decision procedure using Lemma~\ref{lem:modular:nonsep:witness}.
	This one enumerates all linear sets $L$ and $M$
	(not necessarily only those in the special form of $U$- or $V$- witnesses)
	and checks whether they are modular{\sepsep}separable
	\emph{and} included in $U$ and $V$, respectively.
	While this procedure is conceptually simpler than the one we presented,
	we now need the two extra inclusion checks $L \subseteq U$ and $M \subseteq V$.
	Indeed, $U$- and $V$-witnesses were designed in such a way that
	the two inclusions above hold by construction and do not have to be checked.
	The problem whether a given linear set is included in a given \vasreach is decidable~\cite{Leroux:PresburgerVAS:LICS2013},
	however we chose to present the previous semi-decision procedure in order to be self contained.
\end{remark}

\section{Unary{\sepsep}separability of VAS{\sepsep}sections}
\label{sec:unary-vas-sep}

We now embark on the proof of Theorem~\ref{thm:unary-vas-sep}.
It goes along the lines of the proof of Theorem~\ref{thm:modular-vas-sep},
but with some details more complicated, thus we only concentrate on explaining the necessary adjustments.
As before, the positive semi-decision procedure enumerates all $n \in \N$
and checks whether the $\cong_n$-closures of the two reachability sets are disjoint,
which is effective thanks to the following fact:
\begin{lemma}\label{lem:positive-witness}
	For two \vas{\sepsep}sections $U$ and $V$ and $n \in \N$,
	it is decidable whether there exist $u \in U$ and $v \in V$ such that $u \cong_n v$.
\end{lemma}
This can be proved in a way similar to Lemma~\ref{lem:modular-positive-witness},
with the adjustment that we allow on every coordinate a decrement by $n$ only if the value is above $2n$.

The negative semi-decision procedure enumerates nonseparability witnesses, along the same lines as in the case of
modular{\sepsep}separability.
The following crucial lemma is an exact copy of Lemma~\ref{lem:modular:nonsep:witness},
except that ``modular'' is replaced by ``unary'':
%
%
%
\keylemma{unary}{lem:unary:nonsep:witness}
\begin{proof}
	%
	We only concentrate on showing that 2 implies 4.
	Assume that the expansions $U'$ and $V'$ are not unary{\sepsep}separable, for
	two sections $U$ and $V$ represented as (recall the simplifying assumptions about \vas{\sepsep}sections 
	from Section~\ref{sec:modular-vas-sep})
	\[U = \sec {I, 0} {R_U} \subseteq \N^d \quad \textrm{ and } \quad V = \sec {I, 0} {R_V} \subseteq \N^d\enspace,\]
	where $R_U, R_V \subseteq \N^{d'}$ are the reachability sets
	of  two \vases 
	and $I \subseteq \{1, \ldots, d'\}$ 
	with $|I| = d$ are projecting coordinates.
	%
	Since $U'$ and $V'$ are not unary{\sepsep}separable,
	by Proposition~\ref{prop:inseparable-pairs},
	there exists an infinite sequence of pairs of reachable configurations $(u_0, v_0), (u_1, v_1), \ldots \in U' \times V'$
	s.t. $u_n \cong_n v_n$ for all $n \in \N$.
	It means that for every $n \in \N$ there exist runs $\rho_n$ and $\sigma_n$ in the two VASes
	ending up in reachable configurations $u_n := \target{\rho_n} \in R_U$ and $v_n := \target{\sigma_n} \in R_V$.
	%
	Define $\delta_n := u_n - u_0$ and $\gamma_n := v_n - v_0$ for all $n \in \N$.
	Since $\runord$ is a wqo,
	by reasoning as in the proof of Lemma~\ref{lem:modular:nonsep:witness},
	we can assume w.l.o.g. that $\rho_0 \runord \rho_1 \runord \cdots$,
	and similarly for the $\sigma_i$'s.

	Since $u_n \cong_n v_n$,
	the two sequences $u_0 \leq u_1 \leq \cdots$ and $v_0 \leq v_1 \leq \cdots$
	are unbounded on the same set of coordinates.
	Let $F \subseteq \{1, \ldots, d'\}$ be this set; note that $F \subseteq I$ .
	By eliminating a sufficiently long prefix of these two sequences,
	we can further assume that bounded coordinates are in fact constant,
	and again from $u_n \cong_n v_n$ it follows that this constant is the same vector for both sequences.
	Consequently, 
	\begin{align}
		\label{eq:B1}
		&\proj {\set{ 1, \dots, d' } \setminus F} {u_0} = \proj {\set{ 1, \dots, d' } \setminus F} {v_0}\enspace, \textrm{ and } \\
		\label{eq:B2}
		&\prettyforall {n \in \N}
		{\proj {\set{ 1, \dots, d' } \setminus F} {\delta_n} = \proj {\set{ 1, \dots, d' } \setminus F} {\gamma_n} = 0}\enspace.
	\end{align}

	By proceding as in the proof of Lemma~\ref{lem:modular:nonsep:witness},
	there exist two finite sets
	$P := \set{\delta_{i_1}, \dots, \delta_{i_k}}$ and
	$Q := \set{\gamma_{i_1}, \dots, \delta_{i_k}}$ such that
	the linear sets
	$L := \set {u_0} + \lin^{\geq 0} (P) \subseteq U$ is a $U$-witness, the linear set 
	$M := \set {v_0} + \lin^{\geq 0} (Q) \subseteq V$ is a $V$-witness, and $L, M$ are not modular{\sepsep}separable.
	%
	It remains to show that $L$ and $M$ are not \emph{unary} separable either.
	While unary{\sepsep}nonseparability is a stronger property than modular{\sepsep}nonseparability in general,
	by Lemma~\ref{lem:F-linked} the two conditions are in fact equivalent when the two sets are \emph{linked}.
	We make use of the set $F$ as chosen before, and we show that $L$ and $M$ are $F$-linked.
	Indeed, if $j \in F$ then w.l.o.g.~we may assume that the two sequences
	$\proj j {u_0} < \proj j {u_1} < \ldots$ and
	$\proj j {v_0} < \proj j {v_1} < \ldots$ are strictly increasing.
	Thus, $\proj j {\delta_n}, \proj j {\gamma_n} > n$ for every $n \in \N$,
	which implies that $\proj F L$ and $\proj F M$ are diagonal.
	On the other hand, if $j \in \{1, \dots, d' \} \setminus F$,
	from properties~\eqref{eq:B1} and~\eqref{eq:B2} above, we have
	$\proj {\{1, \dots, d\} \setminus F} L = \proj {\{1, \dots, d\} \setminus F} M = \set {\proj {\{1, \dots, d\} \setminus F} {u_0}}$.
	Thus $L$ and $M$ are indeed $F$-linked.
	%
	\qed
\end{proof}



\section{Final remarks}\label{sec:conclusions}

We have shown decidability of modular\sepsep and unary{\sepsep}separability for sections of \vasreach sets,
which include (sections of) reachability sets of \vases with states and Petri nets.
As a corollary, we have derived decidability of regular{\sepsep}separability of commutative closures of \vas languages,
and of commutative regular{\sepsep} separability of \vas languages.
The decidability status of the regular{\sepsep}separability problem for \vas languages remains an intriguing open problem.

\paragraph{Complexity.}
Most of the problems shown decidable in this paper are easily shown to be at least as hard as the \vasreach problem.
In particular, this applies to unary{\sepsep}separability of \vasreach sets,
and to regular{\sepsep}separability of commutative closures of \vaslangs.
Indeed, for unary{\sepsep}separability, it suffices to notice that a configuration $u$ cannot reach a configuration $v$ if, and only if,
the set reachable from $u$ can be unary{\sepsep}separated from the singleton set $\set v$,
also a \vasreach set.
When the separator exists, it can be taken to be the complement of $\set v$ itself,
which is unary.  

While the problem of modular{\sepsep}separability is \textsc{ExpSpace}-hard,
we do not know whether it is as hard as the \vasreach problem.
The hardness can be shown by reduction from the control state reachability problem in \vasses, which is 
\textsc{ExpSpace}-hard~\cite{Lipton76}. For a \vass $V$ and a target control state $q$ thereof,
we construct two new \vasses $V_0$ and $V_1$,
which are copies of $V$ with one additional coordinate,
which at the beginning is zero for $V_0$ and one for $V_1$.
We also add one new transition from control state $q$,
which allows $V_1$ to decrease the additional coordinate by one.
One can easily verify that the two \vass reachability sets definable by $V_0$ and $V_1$ are modular{\sepsep}separable
if, and only if, the control state $q$ is not reachable in $V$, which finishes the proof of \textsc{ExpSpace}-hardness.

\paragraph{The unarity and modularity characterization problems.}
Closely related problems to separability are the modularity and unarity characterization problems:
is a given section of a \vas{\sepsep}reachability set modular, resp., unary?
We focus here on the unarity problem, but the other one can be dealt in the same way.
Decidability of the unarity problem would follow immediately from  Theorem~\ref{thm:unary-vas-sep}, if sections of \vasreach sets were 
(effectively) closed under complement.
This is however not the case.
Indeed, if the complement of a \vasreach set is a section of another \vasreach set,
then both sets are necessarily a section of a Presburger invariant~\cite{DBLP:conf/lics/Leroux09}, hence semilinear.
But we know that \vasreach sets can be non-semilinear,
and thus they are not closed under complement. 
However, the unarity problem can be shown to be decidable directly, at least for \vasreach sets, by using the following two facts: 
first, it is decidable if a given \vasreach set $U$ is semilinear (see the unpublished works~\cite{Hauschildt90,Lambert94});
second, when a \vasreach set is semilinar, a concrete representation thereof as a semilinear set is effectively computable~\cite{L13}.
Indeed, if a given $U$ is not semilinear, it is not unary either;
otherwise, compute a semilinear representation, and check if it is unary.
The latter can be checked directly,
or can be reduced to unary{\sepsep}separability of semilinear sets (since semilinear sets are closed under complement, as discussed above).

\paragraph{Acknowledgements}
We thank Maria Donten-Bury for providing us elegant proofs of Lemmas~\ref{lem:solution-modulo} and \ref{lem:finite-base}.

\bibliographystyle{plain}
\bibliography{citat}

\end{document}